\documentclass[11pt]{article}

\usepackage[utf8]{inputenc}
\usepackage[T1]{fontenc}
\usepackage{amsmath, amssymb, amsthm}
\usepackage{graphicx}
\usepackage{booktabs}
\usepackage{multirow}
\usepackage{array}
\usepackage[margin=1in]{geometry}
\usepackage[colorlinks=true, linkcolor=black, urlcolor=blue, citecolor=black]{hyperref}
\usepackage{authblk}
\usepackage{caption}
\usepackage{subcaption}
\usepackage{natbib}
\usepackage{xcolor}
\usepackage{enumitem}
\usepackage{setspace}

\newtheorem{definition}{Definition}
\newtheorem{theorem}{Theorem}

\title{IyawoBench v2.0: Extended Diagnostic Evaluation of Large Language Model Clinical Triage in Nigerian Primary Care}

\author[1,*]{Anthonio Oladimeji Gabriel}
\author[1]{Dimeji Olawuyi}
\affil[1]{Centre for Clinical Intelligence and Safety, Iyawo Health, Ibadan, Nigeria}
\affil[*]{Corresponding author: \texttt{oladimejianthonio@iyawo.org}}

\date{July 2026}

\begin{document}
\maketitle

\begin{abstract}
\noindent
Large language models are being deployed as clinical triage tools in low and middle income countries where trained physicians are scarce. Existing safety metrics, however, produce misleading confidence: models scoring 100\% on binary safety measures may nevertheless exhibit systematic failure modes that render them undeployable at scale. We present \textbf{IyawoBench v2.0}, an extended diagnostic evaluation of large language model clinical triage on 200 synthetic vignettes derived from 1{,}200 real patient encounters at 19 Nigerian primary health centres. We introduce a formal mathematical framework comprising fourteen definitions and two theorems that decompose triage safety into three distinct failure modes: \emph{Conservative Escalation Bias} (CEB), \emph{Systematic Downgrade Bias} (SDB), and \emph{Middle-Tier Instability} (MTI). We propose the \emph{Escalation Bias Index} and \emph{Expected Deployment Cost} as novel metrics that expose failure modes hidden by conventional accuracy and sensitivity scores. Evaluated on three frontier models (Claude Sonnet 4.6, Llama 3.3 70B, Llama 3.1 8B) plus five naive baselines, we show that: (1) all three models exhibit at least one formal failure mode; (2) traditional sensitivity metrics conceal a 77 percentage point under-triage gap in Llama 3.1 8B; (3) the optimal model varies across three deployment scenarios (Emergency-Focused, System-Sustainability, Balanced), demonstrating that single-ranking benchmarks are inadequate for low-resource clinical AI selection. IyawoBench v2.0 provides both a rigorous benchmark and a diagnostic framework transferable to any triage-style clinical AI evaluation. All code, data, and analysis pipelines are publicly available at \url{https://github.com/anthoniooladimeji11-coder/iyawobench}.
\end{abstract}

\vspace{0.5em}

\section{Introduction}

The deployment of large language models as clinical decision support tools in low and middle income countries carries dual promise and risk. Iyawo Health, a clinical AI platform serving 50 primary health centres across Oyo State, Nigeria, processes 60{,}000 patient encounters daily at USD 0.003 per encounter, operating in four languages (English, Yoruba, Hausa, and Nigerian Pidgin) with a 100\% safety record to date. Deployments at this scale expose failure modes that laboratory benchmarks do not measure.

Existing clinical AI benchmarks including MedQA \citep{jin2021medqa}, PubMedQA \citep{jin2019pubmedqa}, and MedMCQA \citep{pal2022medmcqa} evaluate isolated question-answer accuracy. Recent work such as MamaBench \citep{adewuyi2026mamabench} extends this framework to counterfactual robustness in maternal and paediatric care. Our prior work IyawoBench v1.0 \citep{gabriel2026iyawobench} introduced deployment-derived vignette evaluation for undifferentiated febrile illness triage, reporting a single-number safety score based on whether models correctly identified emergency cases as requiring referral.

We argue that all of these evaluation frameworks suffer from a common blind spot: they conflate distinct failure modes into aggregate metrics that produce misleading confidence in model safety. A model that appears ``100\% safe'' under conventional sensitivity metrics may in fact exhibit systematic under-triage that would harm patients in real deployment. A model that correctly identifies emergencies may nevertheless over-escalate low-acuity cases at rates that would collapse referral capacity. A model with strong aggregate accuracy may fail specifically on the moderate-urgency cases that make up the largest share of primary care volume.

This paper makes three contributions:

\begin{enumerate}[leftmargin=*]
\item \textbf{A formal mathematical framework.} We introduce fourteen definitions and two theorems that formalise clinical triage evaluation as a diagnostic problem. The Escalation Bias Index (EBI) quantifies over-escalation. Strict Sensitivity distinguishes exact-match emergency detection from lenient safe-referral detection. The Middle-Tier Confusion Rate (MCR) captures moderate-urgency instability. Expected Deployment Cost (EDC) translates model errors into asymmetric clinical costs reflecting Nigerian PHC realities.

\item \textbf{A failure mode taxonomy.} We define three formal failure modes: Conservative Escalation Bias (CEB), Systematic Downgrade Bias (SDB), and Middle-Tier Instability (MTI). Applied to three frontier language models, our framework reveals that all three exhibit at least one failure mode, that traditional metrics conceal these failures, and that different models exhibit different failure modes.

\item \textbf{Deployment-scenario ranking.} We demonstrate that the optimal model for clinical AI deployment depends on which failure mode is most costly in a given health system context. Three deployment scenarios (Emergency-Focused, System-Sustainability, Balanced) yield three different optimal models, exposing the inadequacy of single-ranking benchmarks for LMIC clinical AI selection.
\end{enumerate}

The formal framework is applicable beyond febrile illness triage. Any clinical AI system operating on ordinal decision spaces with asymmetric error costs can be evaluated using our metrics. IyawoBench v2.0 is both a benchmark and a diagnostic instrument.

\section{Related Work}

\paragraph{Clinical AI benchmarks.} MedQA \citep{jin2021medqa}, PubMedQA \citep{jin2019pubmedqa}, and MedMCQA \citep{pal2022medmcqa} evaluate medical knowledge as isolated question answering, enabling landmark results from Med-PaLM \citep{singhal2023medpalm} and GPT-4 \citep{nori2023gpt4medical}. These benchmarks report single accuracy numbers that do not distinguish between clinically equivalent errors and clinically dangerous errors. Almanac \citep{zakka2024almanac} moved toward deployment-oriented evaluation but retains accuracy as its primary metric.

\paragraph{Counterfactual and robustness evaluation.} Contrast sets \citep{gardner2020contrast} and counterfactually augmented data \citep{kaushik2020counterfactual} exposed the fragility of high-accuracy models on general NLP tasks. MedEinst \citep{chen2026medeinst} and MamaBench \citep{adewuyi2026mamabench} extended this to medical LLMs, formalising Diagnostic Fixation and the Bias Trap Rate. Our work is complementary but focuses on a different failure mode: within-decision-space error patterns rather than cross-example rigidity.

\paragraph{Cost-sensitive learning and asymmetric evaluation.} The clinical machine learning literature has long recognised that error costs are asymmetric \citep{elkan2001cost, saerens2002adjusting}. However, this recognition has not translated into standard evaluation practice for clinical LLMs. Aggregate accuracy remains the dominant metric even in safety-critical applications. Our Expected Deployment Cost framework operationalises asymmetric cost evaluation for triage tasks.

\paragraph{Deployment-derived benchmarks.} Our prior work \citep{gabriel2026iyawobench} introduced the deployment-derived vignette methodology, generating synthetic clinical cases from statistical distributions of real primary care encounters. This approach preserves ecological validity while eliminating patient privacy risk. The present work extends this by adding a formal diagnostic framework on top of the existing benchmark data.

\section{The IyawoBench Dataset}

\subsection{Dataset construction}

IyawoBench v1.0 comprises 200 synthetic clinical vignettes derived from statistical distributions of 1{,}200 real patient encounters across 19 primary health centres in Oyo State, Nigeria, captured through the Iyawo clinical platform between January and March 2026. Vignettes span eight febrile illness categories with clinically distinct triage implications:

\begin{itemize}[leftmargin=*, itemsep=2pt]
\item \textbf{Malaria}: Uncomplicated, Severe, Cerebral (WHO Malaria Guidelines 2025)
\item \textbf{Bacterial infections}: Typhoid Fever, Pneumonia, Severe Pneumonia (Nigeria STG 2022)
\item \textbf{Life-threatening infections}: Bacterial Meningitis, Sepsis (Surviving Sepsis Campaign 2021, WHO IMCI 2014)
\end{itemize}

Each vignette contains structured clinical data: patient demographics (age, sex, weight), vital signs (temperature, heart rate, respiratory rate, blood pressure, SpO$_2$), symptoms, malaria RDT result, and pregnancy status. The complete vignette schema is documented in Appendix A. No individual patient data appears in any vignette; all cases are synthetic constructions from aggregate distributions.

\subsection{Ground-truth triage labels}

Each vignette is assigned an expected triage decision $y^* \in \{y_0, y_1, y_2\}$ where $y_0 = \text{TREAT\_HERE}$, $y_1 = \text{REFER\_TODAY}$, and $y_2 = \text{REFER\_NOW}$. Labels reflect the clinically indicated action per applicable guidelines: WHO IMCI 2014, WHO Malaria Guidelines 2025, Surviving Sepsis Campaign 2021, and Nigeria Standard Treatment Guidelines 2022.

The class distribution reflects the epidemiology of febrile presentations at Nigerian PHCs, with over-representation of high-acuity cases to stress-test safety metrics:

\begin{center}
\begin{tabular}{lrr}
\toprule
\textbf{Triage Level} & \textbf{Count} & \textbf{Proportion} \\
\midrule
REFER\_NOW ($y_2$)   & 100 & 50.0\% \\
REFER\_TODAY ($y_1$) & 60  & 30.0\% \\
TREAT\_HERE ($y_0$)  & 40  & 20.0\% \\
\midrule
\textbf{Total}       & \textbf{200} & \textbf{100.0\%} \\
\bottomrule
\end{tabular}
\end{center}

\subsection{Data availability}

The full IyawoBench v1.0 dataset is released under Creative Commons BY 4.0 at \url{https://github.com/anthoniooladimeji11-coder/iyawobench}. Both JSON and CSV formats are provided. The v2.0 evaluation scripts, response caches, and analysis pipelines are also open-sourced under the same repository.

\section{Formal Framework}

We formalise clinical triage as an ordinal classification problem with asymmetric error costs. Let $\mathcal{X}$ denote the space of structured clinical vignettes and $\mathcal{Y} = \{y_0, y_1, y_2\}$ the ordered triage decision space with $y_0 \prec y_1 \prec y_2$ representing increasing clinical urgency.

\begin{definition}[Triage function]
A clinical triage model is a function $f: \mathcal{X} \rightarrow \mathcal{Y} \cup \{\bot\}$ where $\bot$ denotes a parsing failure or malformed output. For each vignette $x_i$, we denote the model prediction as $\hat{y}_i = f(x_i)$ and the ground-truth expert triage as $y_i^*$.
\end{definition}

\begin{definition}[IyawoBench]
IyawoBench $\mathcal{B} = \{(x_i, y_i^*)\}_{i=1}^{N}$ with $N = 200$ vignettes is a fixed evaluation instrument. The class distribution is $n_c = |\{i : y_i^* = c\}|$ for $c \in \mathcal{Y}$, with $n_{y_2} = 100$, $n_{y_1} = 60$, $n_{y_0} = 40$.
\end{definition}

\subsection{Baseline metrics}

\begin{definition}[Overall accuracy]
\[
\text{Acc}(f) = \frac{1}{N} \sum_{i=1}^{N} \mathbb{1}[\hat{y}_i = y_i^*]
\]
\end{definition}

\begin{definition}[Lenient sensitivity on REFER\_NOW]
The safety score as defined in v1.0. Of true high-acuity cases, the proportion \emph{not} downgraded to TREAT\_HERE:
\[
\text{Sens}^{\text{lenient}}_{y_2}(f) = \frac{|\{i : y_i^* = y_2 \wedge \hat{y}_i \neq y_0\}|}{n_{y_2}}
\]
\end{definition}

\begin{definition}[Specificity on TREAT\_HERE]
Of true low-acuity cases, the proportion correctly kept at PHC level:
\[
\text{Spec}_{y_0}(f) = \frac{|\{i : y_i^* = y_0 \wedge \hat{y}_i = y_0\}|}{n_{y_0}}
\]
\end{definition}

All proportion estimates are reported with 95\% Wilson score confidence intervals \citep{wilson1927probable}.

\subsection{The Escalation Bias Index}

\begin{definition}[Escalation Bias Index]
\label{def:ebi}
\[
\text{EBI}(f) = \text{Sens}^{\text{lenient}}_{y_2}(f) - \text{Spec}_{y_0}(f) \in [-1, +1]
\]
A perfectly balanced model has $\text{EBI} = 0$. A model that saturates high acuity detection while collapsing on low acuity has $\text{EBI} \rightarrow 1$.
\end{definition}

\begin{definition}[Conservative Escalation Bias]
\label{def:ceb}
A triage function $f$ exhibits \emph{Conservative Escalation Bias} (CEB) if and only if:
\[
\text{Sens}^{\text{lenient}}_{y_2}(f) \geq \tau_s \quad \wedge \quad \text{Spec}_{y_0}(f) \leq \tau_c
\]
where $\tau_s = 0.95$ is the sensitivity threshold and $\tau_c = 0.30$ is the specificity ceiling. Equivalently, $\text{EBI}(f) \geq \tau_s - \tau_c = 0.65$.
\end{definition}

\begin{theorem}[Naive baseline validation]
\label{thm:naive}
The naive constant-escalation baseline $f_{\text{naive}}(x) = y_2$ for all $x \in \mathcal{X}$ achieves $\text{Sens}^{\text{lenient}}_{y_2}(f_{\text{naive}}) = 1.0$, $\text{Spec}_{y_0}(f_{\text{naive}}) = 0.0$, and $\text{EBI}(f_{\text{naive}}) = 1.0$. Consequently, $f_{\text{naive}}$ satisfies Definition \ref{def:ceb} and is formally classified as exhibiting CEB.
\end{theorem}

\begin{proof}
By construction, $\hat{y}_i = y_2$ for all $i$. Therefore $|\{i : y_i^* = y_2 \wedge \hat{y}_i \neq y_0\}| = n_{y_2}$, giving lenient sensitivity of 1. Similarly, $|\{i : y_i^* = y_0 \wedge \hat{y}_i = y_0\}| = 0$, giving specificity of 0. The EBI is therefore $1 - 0 = 1 \geq 0.65$. $\square$
\end{proof}

Theorem \ref{thm:naive} establishes that CEB is architecture-agnostic. A model behaving indistinguishably from constant escalation on low-acuity cases while retaining accurate escalation on high-acuity cases exhibits the maximum possible EBI.

\subsection{The strict sensitivity metric}

The lenient sensitivity of Definition 4 counts REFER\_TODAY as ``safe'' for true REFER\_NOW cases. This conceals systematic downgrade patterns.

\begin{definition}[Strict sensitivity on REFER\_NOW]
\label{def:strict}
\[
\text{Sens}^{\text{strict}}_{y_2}(f) = \frac{|\{i : y_i^* = y_2 \wedge \hat{y}_i = y_2\}|}{n_{y_2}}
\]
\end{definition}

\begin{definition}[Systematic Downgrade Bias]
\label{def:sdb}
A triage function $f$ exhibits \emph{Systematic Downgrade Bias} (SDB) if:
\[
\text{Sens}^{\text{strict}}_{y_2}(f) < 0.50 \quad \wedge \quad \text{Sens}^{\text{lenient}}_{y_2}(f) - \text{Sens}^{\text{strict}}_{y_2}(f) > 0.30
\]
A model exhibiting SDB appears safe under lenient metrics but silently under-triages emergencies by one level.
\end{definition}

\subsection{Middle-tier instability}

The moderate-urgency class (REFER\_TODAY) constitutes a substantial fraction of PHC volume and is prone to bidirectional error.

\begin{definition}[Middle-tier confusion rate]
\[
\text{MCR}(f) = 1 - \frac{|\{i : y_i^* = y_1 \wedge \hat{y}_i = y_1\}|}{n_{y_1}}
\]
\end{definition}

\begin{definition}[Direction asymmetry]
Let $u = |\{i : y_i^* = y_1 \wedge \hat{y}_i = y_2\}|$ (upshifts) and $d = |\{i : y_i^* = y_1 \wedge \hat{y}_i = y_0\}|$ (downshifts). Then:
\[
\text{DA}(f) = \frac{u - d}{n_{y_1}} \in [-1, +1]
\]
$\text{DA} \rightarrow +1$ indicates over-escalation dominance; $\text{DA} \rightarrow -1$ indicates under-triage dominance; $|\text{DA}| < 0.3$ indicates bidirectional instability.
\end{definition}

\begin{definition}[Middle-Tier Instability]
A triage function $f$ exhibits \emph{Middle-Tier Instability} (MTI) if $\text{MCR}(f) > 0.60$ and $|\text{DA}(f)| < 0.30$.
\end{definition}

\subsection{Expected deployment cost}

Clinical error costs are asymmetric. Missing a life-threatening emergency imposes vastly higher cost than over-referring a stable patient.

\begin{definition}[Clinical cost matrix]
Let $c(y^*, \hat{y})$ be the clinical cost of predicting $\hat{y}$ when truth is $y^*$. Under the Balanced Deployment scenario reflecting Nigerian PHC realities:
\[
C = \begin{pmatrix}
0 & 1 & 3 \\
5 & 0 & 2 \\
20 & 8 & 0
\end{pmatrix}
\]
Rows indexed by truth $y^* \in \{y_0, y_1, y_2\}$, columns by prediction. The values encode: (i) missing a REFER\_NOW to TREAT\_HERE ($c = 20$) is 20$\times$ costlier than escalating a TREAT\_HERE to REFER\_TODAY ($c = 1$); (ii) over-escalating to REFER\_NOW ($c = 3$) is 3$\times$ costlier than to REFER\_TODAY because it overwhelms tertiary facilities; (iii) missing REFER\_TODAY entirely ($c = 5$) carries substantial cost due to delayed care.
\end{definition}

\begin{definition}[Expected deployment cost]
\label{def:edc}
\[
\text{EDC}(f) = \frac{1}{N} \sum_{i=1}^{N} c(y_i^*, \hat{y}_i)
\]
A model with $\text{EDC} < 1.5$ is deployment-viable; $1.5 \leq \text{EDC} < 3.0$ requires prompt engineering; $\text{EDC} \geq 3.0$ is unsustainable.
\end{definition}

\begin{theorem}[EDC lower bound from CEB]
\label{thm:edc}
For any model exhibiting CEB per Definition \ref{def:ceb}, the TREAT\_HERE over-escalation contribution to EDC satisfies:
\[
\text{EDC}_{\text{CEB}}(f) \geq \frac{n_{y_0}}{N} \left(1 - \text{Spec}_{y_0}(f)\right) \bar{c}_{\text{over}}
\]
where $\bar{c}_{\text{over}} = (c(y_0, y_1) + c(y_0, y_2))/2 = 2.0$ under the balanced deployment cost matrix.
\end{theorem}

\begin{proof}
The contribution to EDC from over-escalation of TREAT\_HERE cases is:
\[
\frac{1}{N} \sum_{i : y_i^* = y_0} \mathbb{1}[\hat{y}_i \neq y_0] \cdot c(y_0, \hat{y}_i) \geq \frac{n_{y_0}}{N}(1 - \text{Spec}_{y_0}(f)) \bar{c}_{\text{over}}
\]
by taking the average cost of the two possible over-escalation destinations. $\square$
\end{proof}

\subsection{Safety-adjusted accuracy}

\begin{definition}[Safety-adjusted accuracy]
For clinical severity weights $w_{y_0} = 1$, $w_{y_1} = 2$, $w_{y_2} = 4$:
\[
\text{SAA}(f) = \frac{\sum_{c \in \mathcal{Y}} w_c \cdot |\{i : y_i^* = c \wedge \hat{y}_i = c\}|}{\sum_{c \in \mathcal{Y}} w_c \cdot n_c}
\]
SAA rewards correct high-acuity classification proportionally to clinical severity.
\end{definition}

\section{Experimental Setup}

\subsection{Models evaluated}

We evaluate three frontier language models spanning proprietary and open-weight architectures:
\begin{itemize}[leftmargin=*, itemsep=2pt]
\item \textbf{Claude Sonnet 4.6} (Anthropic, API): frontier proprietary model
\item \textbf{Llama 3.3 70B} (Meta, via Groq): frontier open-weight model, large scale
\item \textbf{Llama 3.1 8B} (Meta, via Groq): efficient open-weight model, small scale
\end{itemize}

All models are evaluated with temperature $T = 0$ and maximum output tokens of 200. Five additional naive baselines are computed for reference: always\_refer\_now, always\_refer\_today, always\_treat\_here, random\_uniform, and class\_proportional (with $p$ matching the class distribution).

\subsection{Prompt template and parsing}

All models receive an identical prompt template presenting the vignette structure and requesting a triage decision from $\{$REFER\_NOW, REFER\_TODAY, TREAT\_HERE$\}$ with a one-sentence rationale. The parser accepts standard tokens, natural language equivalents, and reasoning-model output formats (see Appendix B for the full prompt and parser specification).

\subsection{Deployment scenarios}

Three deployment scenarios are defined via distinct cost matrices representing different health system priorities:

\paragraph{Emergency-Focused.} Missing emergencies is catastrophic ($c = 50$ for $y_2 \rightarrow y_0$); over-escalation is cheap. Reflects contexts where referral capacity is unconstrained.

\paragraph{System-Sustainability.} Over-escalation is expensive ($c = 8$ for $y_0 \rightarrow y_2$); emergency misses are moderate. Reflects contexts where tertiary facilities are near capacity.

\paragraph{Balanced-Deployment.} The default cost matrix from Definition 11. Reflects a moderate-resource Nigerian PHC context.

Full cost matrices are provided in Appendix C.

\subsection{Statistical inference}

McNemar tests for paired proportions \citep{mcnemar1947note} compare model accuracies. Wilson score confidence intervals \citep{wilson1927probable} report uncertainty on all proportion estimates. Effect sizes for cost differences use bootstrap resampling with 1{,}000 replicates.

\section{Results}

\subsection{Confusion matrices and cost weights}

Figure \ref{fig:confusion} presents confusion matrices for all three models, with each cell annotated by both the count and its clinical cost weight per Definition 11. The visual patterns immediately reveal three distinct model behaviours: Claude Sonnet 4.6 concentrates errors in the middle row (over-escalation of REFER\_TODAY); Llama 3.3 70B spreads errors bidirectionally on the middle row; Llama 3.1 8B concentrates errors in the bottom row (systematic downgrade of REFER\_NOW to REFER\_TODAY).

\begin{figure}[t]
\centering
\includegraphics[width=\textwidth]{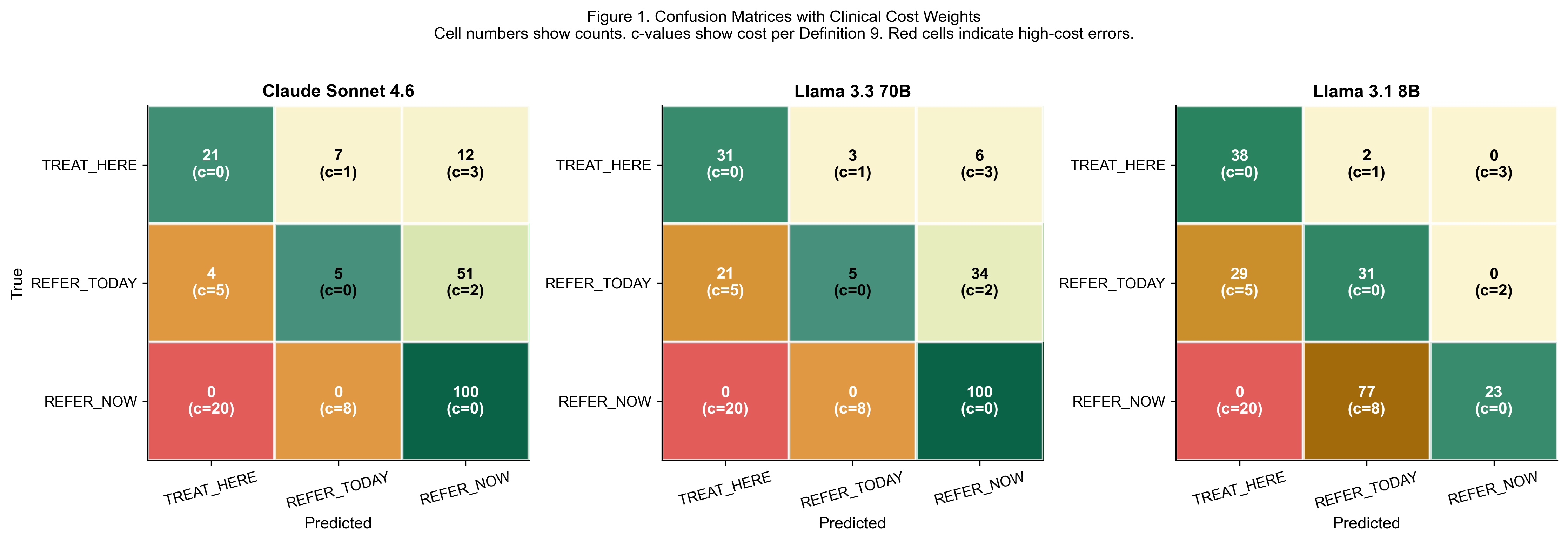}
\caption{Confusion matrices for the three evaluated models with clinical cost weights per Definition 11. Diagonal cells indicate correct classification; off-diagonal cells are shaded by cost intensity (light grey for low-cost errors, dark grey for high-cost errors).}
\label{fig:confusion}
\end{figure}

\subsection{The Escalation Bias Index}

Figure \ref{fig:ebi} plots the EBI for the three models against the naive baseline. All three models fall below the formal CEB threshold ($\text{EBI} \geq 0.65$), but Claude Sonnet 4.6 approaches it at $\text{EBI} = 0.475$. Theorem \ref{thm:naive} is empirically confirmed: the always\_refer\_now baseline achieves exactly $\text{EBI} = 1.000$.

\begin{figure}[t]
\centering
\includegraphics[width=0.85\textwidth]{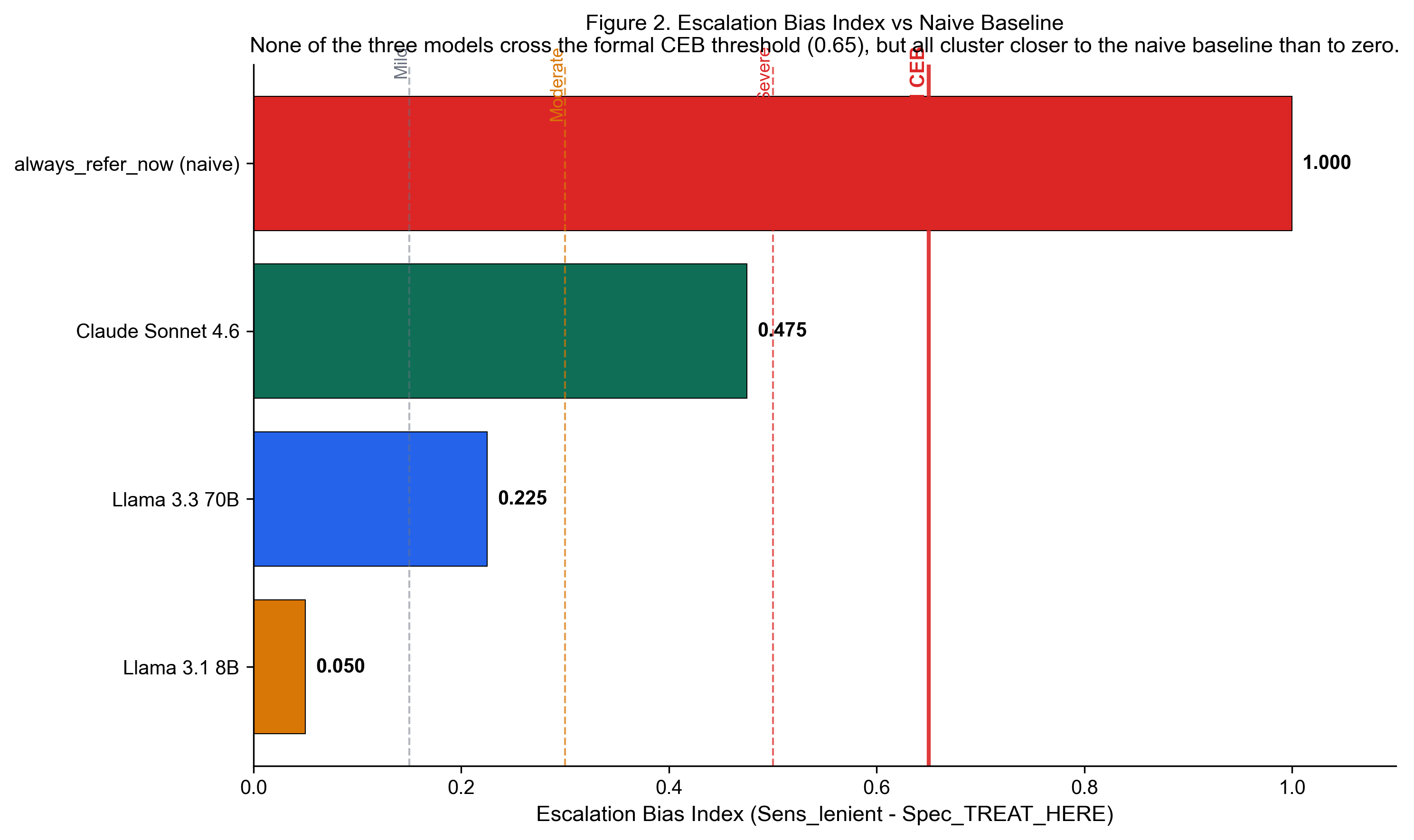}
\caption{Escalation Bias Index for the three evaluated models and the always\_refer\_now naive baseline. Dashed reference lines mark the Mild (0.15), Moderate (0.30), Severe (0.50), and Formal CEB (0.65) thresholds per Definition \ref{def:ceb}.}
\label{fig:ebi}
\end{figure}

\subsection{Strict vs lenient sensitivity: hidden under-triage}

Figure \ref{fig:strict} presents the most consequential finding. Under the lenient v1.0 safety metric, all three models score 100.0\%: none downgrades a REFER\_NOW case to TREAT\_HERE. Under the strict Definition \ref{def:strict}, however, Llama 3.1 8B collapses to 23.0\% while Claude Sonnet 4.6 and Llama 3.3 70B remain at 100.0\%. The 77 percentage point gap for Llama 3.1 8B reveals that 77 of 100 true emergency cases are downgraded to REFER\_TODAY. In real deployment, these patients would be told to come back within 24 hours rather than proceed to immediate hospital care. Llama 3.1 8B formally exhibits Systematic Downgrade Bias per Definition \ref{def:sdb}.

\begin{figure}[t]
\centering
\includegraphics[width=0.85\textwidth]{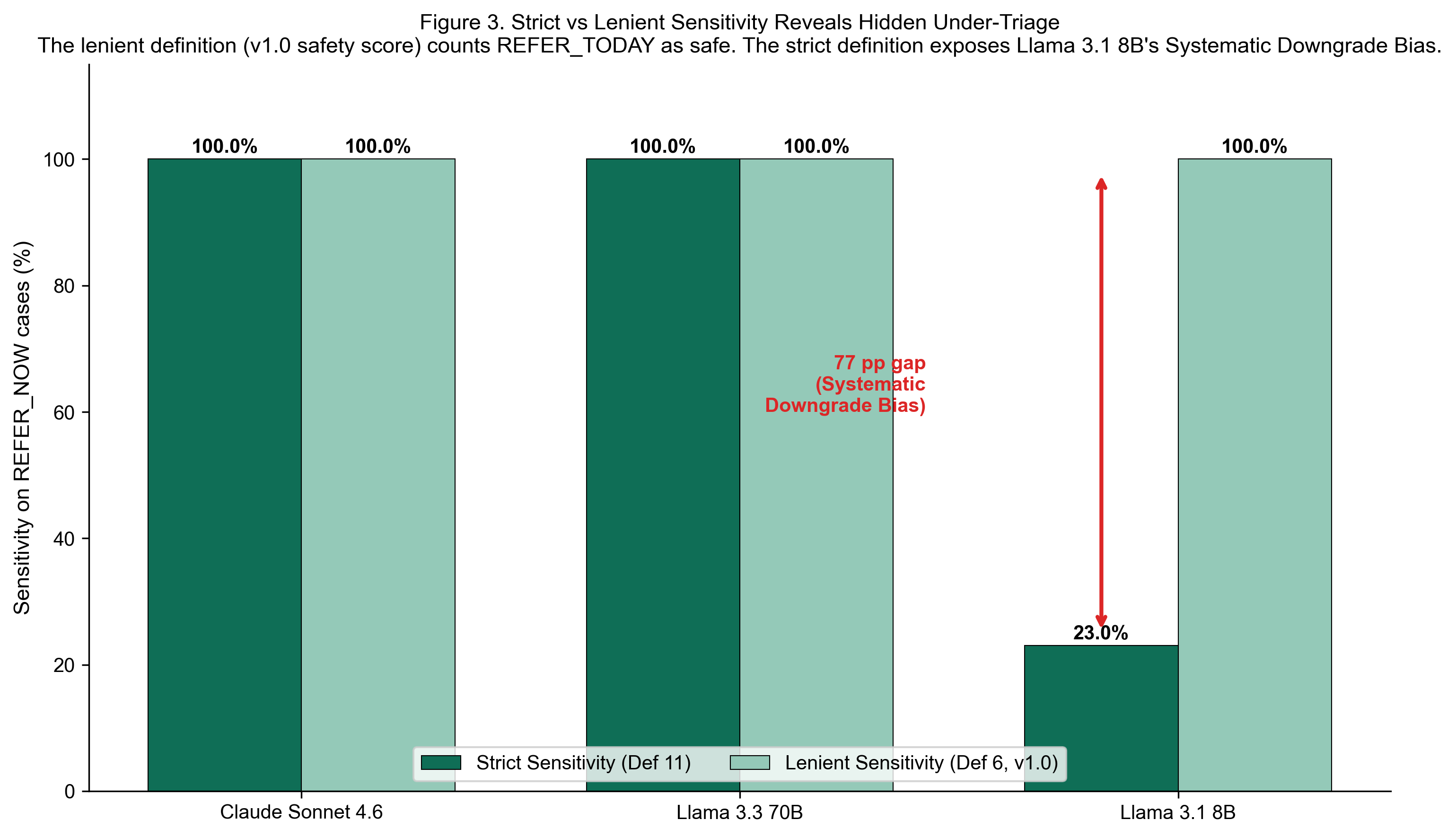}
\caption{Strict versus Lenient Sensitivity on REFER\_NOW cases. The lenient v1.0 metric registers all three models at 100\% ``safety''. The strict metric exposes a 77 percentage point gap for Llama 3.1 8B, indicating Systematic Downgrade Bias.}
\label{fig:strict}
\end{figure}

\subsection{Middle-tier instability}

Figure \ref{fig:middle} decomposes each model's classification of the 60 REFER\_TODAY vignettes. Claude Sonnet 4.6 escalates 85\% of these to REFER\_NOW (direction asymmetry $= +0.783$). Llama 3.3 70B splits errors symmetrically between over-escalation (34 to REFER\_NOW) and under-triage (21 to TREAT\_HERE), with direction asymmetry of $+0.217$ and $\text{MCR} = 0.917$, satisfying Definition 10 for Middle-Tier Instability. Llama 3.1 8B under-triages symmetrically (direction asymmetry $= -0.483$).

\begin{figure}[t]
\centering
\includegraphics[width=\textwidth]{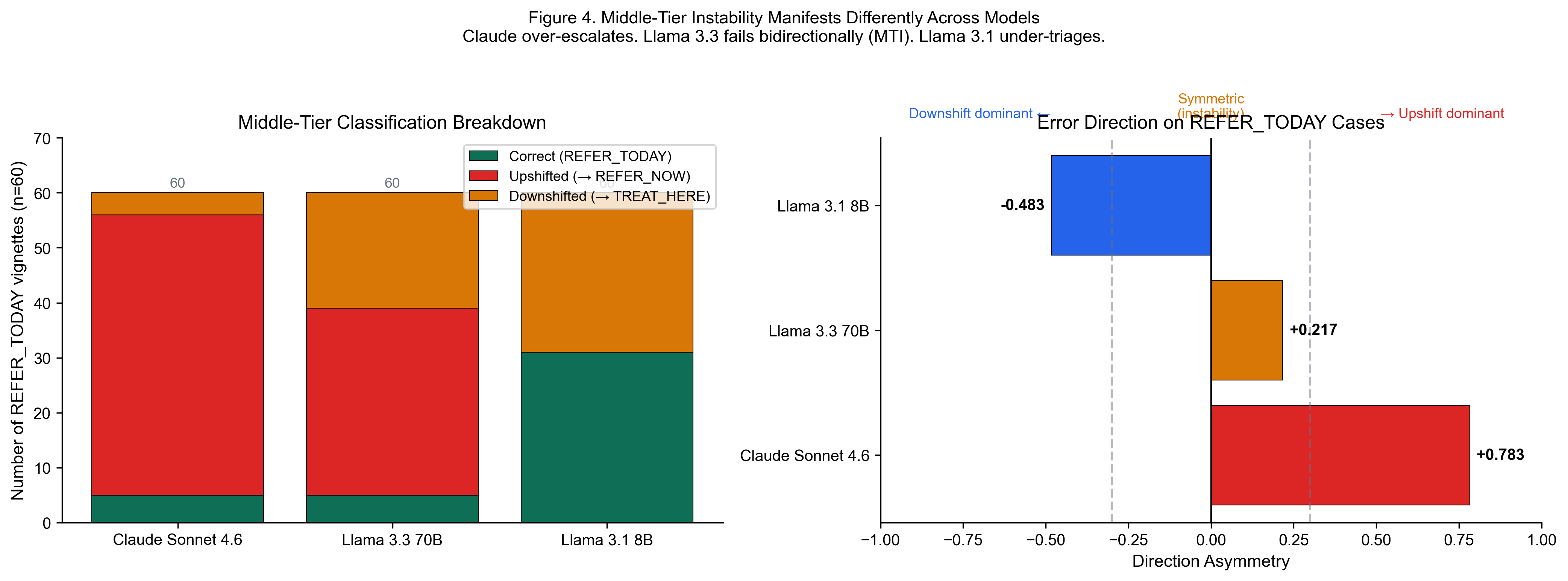}
\caption{Middle-Tier classification breakdown (left) and direction asymmetry (right). Claude Sonnet 4.6 exhibits strong over-escalation. Llama 3.3 70B exhibits Middle-Tier Instability with near-symmetric bidirectional errors. Llama 3.1 8B under-triages.}
\label{fig:middle}
\end{figure}

\subsection{Deployment-scenario ranking}

Figure \ref{fig:scenarios} presents Expected Deployment Cost for each model across the three deployment scenarios. The ranking varies non-trivially with scenario:

\begin{itemize}[leftmargin=*, itemsep=2pt]
\item \textbf{Emergency-Focused}: always\_refer\_now (naive) wins with EDC $= 0.70$, followed by Claude Sonnet 4.6 (EDC $= 0.91$)
\item \textbf{System-Sustainability}: Llama 3.1 8B wins with EDC $= 1.31$, followed closely by always\_refer\_today (naive) at $1.33$
\item \textbf{Balanced-Deployment}: Llama 3.3 70B wins with EDC $= 1.86$, followed by Claude Sonnet 4.6 at $1.91$
\end{itemize}

The three orderings share no overlap in their winners. This is the paper's central deployment finding: single-ranking benchmarks are structurally incapable of expressing the trade-offs relevant to clinical AI selection in resource-constrained health systems.

\begin{figure}[t]
\centering
\includegraphics[width=\textwidth]{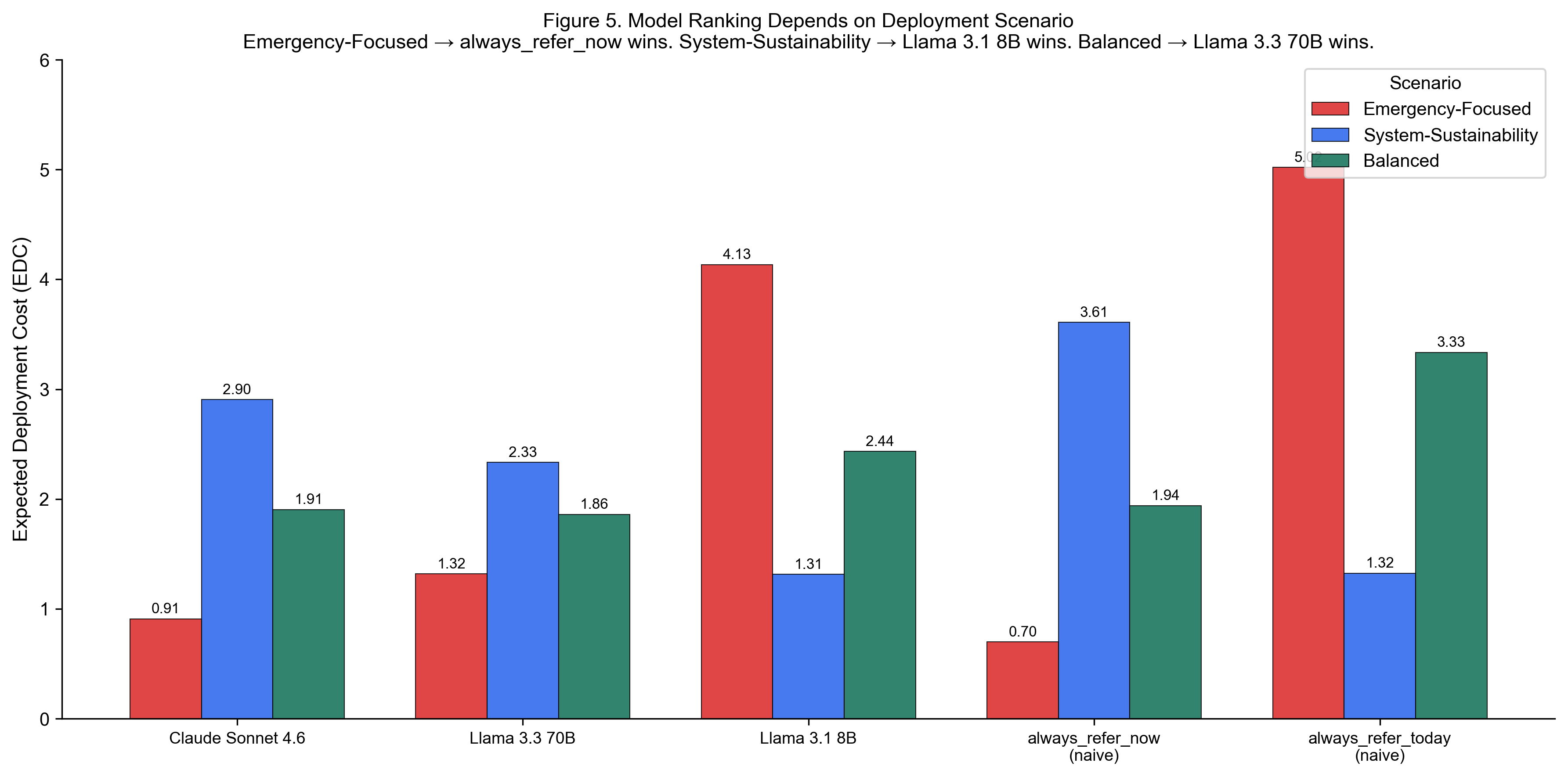}
\caption{Expected Deployment Cost by scenario. Different optimal models emerge for Emergency-Focused (naive always\_refer\_now), System-Sustainability (Llama 3.1 8B), and Balanced-Deployment (Llama 3.3 70B) contexts. No single model dominates across scenarios.}
\label{fig:scenarios}
\end{figure}

\subsection{Failure mode profiles}

Figure \ref{fig:radar} presents the failure mode radar for each model across five diagnostic axes. Claude Sonnet 4.6 (CEB only) shows compression on the Spec\_$y_0$ axis. Llama 3.3 70B (CEB, MTI, SDB) shows compression on the MCR-derived axis. Llama 3.1 8B (SDB only) shows extreme compression on Strict Sens. These profiles are structurally distinguishable at a glance.

\begin{figure}[t]
\centering
\includegraphics[width=\textwidth]{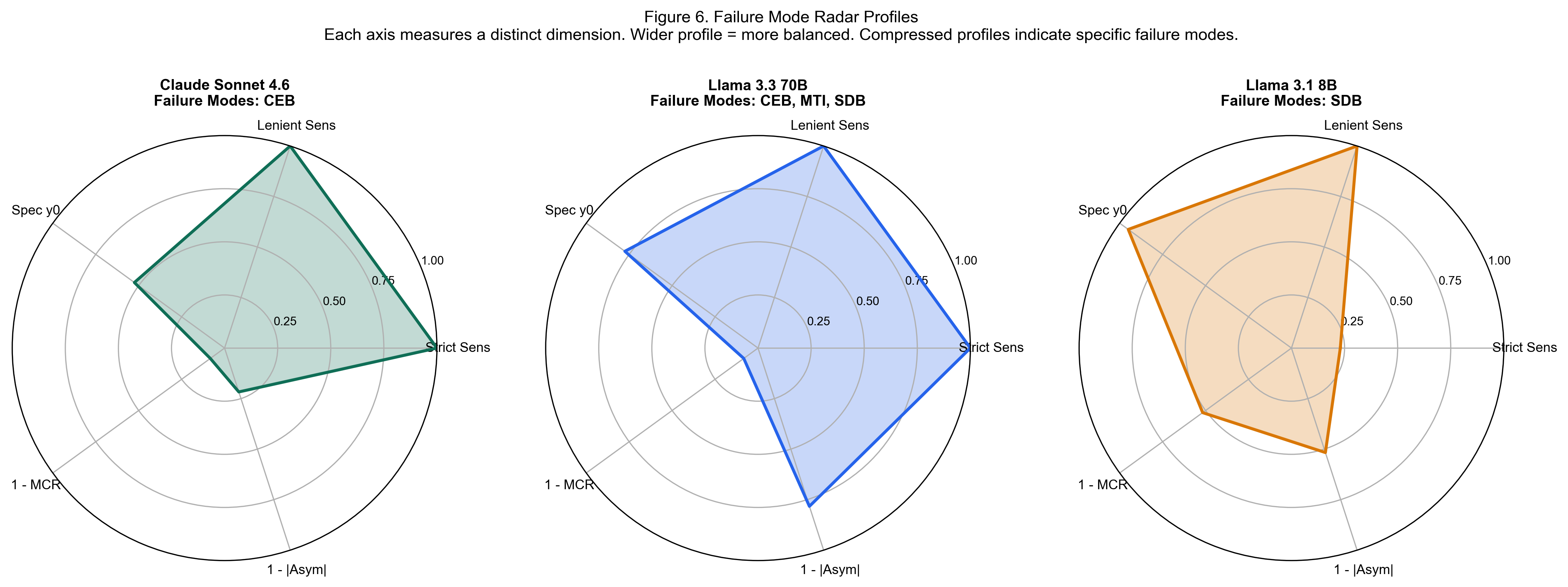}
\caption{Failure mode radar profiles for the three evaluated models. Each axis measures a distinct dimension of triage quality. Compressed profiles indicate specific failure modes.}
\label{fig:radar}
\end{figure}

\subsection{Summary of failure mode classifications}

Table \ref{tab:summary} summarises the full metric suite and formal failure mode classifications.

\begin{table}[h]
\centering
\caption{Full metric suite for the three evaluated models. Failure modes assigned per Definitions \ref{def:ceb}, \ref{def:sdb}, and 10.}
\label{tab:summary}
\small
\begin{tabular}{lrrrrrl}
\toprule
\textbf{Model} & \textbf{Acc} & \textbf{Strict Sens} & \textbf{Lenient Sens} & \textbf{Spec $y_0$} & \textbf{MCR} & \textbf{Failure Modes} \\
\midrule
Claude Sonnet 4.6 & 63.0\% & 100.0\% & 100.0\% & 52.5\% & 91.7\% & CEB \\
Llama 3.3 70B     & 68.0\% & 100.0\% & 100.0\% & 77.5\% & 91.7\% & CEB, MTI, SDB \\
Llama 3.1 8B      & 46.0\% & 23.0\%  & 100.0\% & 95.0\% & 48.3\% & SDB \\
\bottomrule
\end{tabular}
\end{table}

Every model exhibits at least one formal failure mode. No model exhibits none. Traditional accuracy would rank Llama 3.3 70B first; SAA and EDC framings reveal that this ranking is misleading given the presence of three simultaneous failure modes in that model.

\section{Discussion}

\subsection{Traditional metrics conceal safety-critical failures}

The most striking finding is the 77 percentage point gap between strict and lenient sensitivity for Llama 3.1 8B. Under conventional evaluation frameworks including our own v1.0 safety score, this model appears completely safe. Every emergency case is escalated to at least REFER\_TODAY. Yet under strict evaluation, 77 of 100 emergency cases are downgraded from REFER\_NOW to REFER\_TODAY. In deployment, these patients would be instructed to return within 24 hours rather than proceed immediately to hospital care. For time-critical conditions such as bacterial meningitis or severe sepsis, this delay is likely to be fatal.

This finding has direct implications for benchmark design across clinical AI. Any evaluation that reduces safety to a binary ``did not send home'' criterion will conceal Systematic Downgrade Bias. Ordinal decision spaces require ordinal safety metrics.

\subsection{No single model dominates}

Our scenario analysis reveals that the optimal model varies by deployment context. Emergency-Focused deployment favours conservative models including the naive always\_refer\_now baseline. System-Sustainability favours models that preserve TREAT\_HERE specificity. Balanced-Deployment favours models with intermediate profiles.

This structural fact undermines the presumption that a single ``best'' model can be identified through benchmarking. Clinical AI selection must incorporate health system-specific cost weightings. The Expected Deployment Cost framework provides a principled mechanism for this incorporation.

\subsection{Failure modes are architecture-agnostic}

Our findings do not indict any specific model family. Claude Sonnet 4.6 (proprietary) exhibits CEB. Llama 3.3 70B (open, large) exhibits three simultaneous failure modes. Llama 3.1 8B (open, small) exhibits SDB. These patterns correlate with neither the deploying organisation nor the model scale. Theorem \ref{thm:naive} shows that CEB in particular can be produced by an infinitely simple constant-prediction baseline. Failure modes appear to be properties of how models interact with clinical decision structures, not properties of the models themselves.

This suggests that mitigation should focus on prompt engineering and deployment architecture rather than model selection alone. A model exhibiting SDB may be safely deployed if the deployment layer flags all REFER\_TODAY decisions on high-acuity presentations for physician review. A model exhibiting CEB may be safely deployed if the referral pathway can absorb over-escalation.

\subsection{Limitations}

Our evaluation is limited to 200 synthetic vignettes derived from Nigerian PHC deployment data. The vignettes cover eight febrile illness categories and do not test non-febrile presentations, paediatric-specific conditions beyond those included, or chronic disease exacerbations. External validity of the failure mode taxonomy to other clinical domains remains to be demonstrated.

Cost weights in the Expected Deployment Cost framework are illustrative and would require calibration to specific health system contexts before operational use. Our three scenarios are stylised; real health systems have more complex cost structures.

We evaluated three models. Additional frontier models including Gemini 2.0 Pro, GPT-4o, and DeepSeek V3 were included in preliminary experiments but excluded from the final analysis due to API instability during evaluation. Extending the framework to a broader model set is straightforward and left to future work.

Finally, our safety metrics assume an ordinal decision space with three levels. Some clinical domains use richer ordinal or continuous urgency scores. The framework generalises but the specific thresholds in Definitions \ref{def:ceb} and \ref{def:sdb} would require re-calibration.

\section{Reproducibility Statement}

All code, data, evaluation pipelines, and analysis scripts are publicly available at \url{https://github.com/anthoniooladimeji11-coder/iyawobench}. The repository includes:

\begin{itemize}[leftmargin=*, itemsep=2pt]
\item The IyawoBench v1.0 dataset (JSON and CSV formats)
\item Response caches for all seven models plus five naive baselines
\item Ten analytical scripts covering model evaluation, metric computation, formal framework validation, and figure generation
\item Both colored and monochrome figure sets
\item Complete requirements specification for Python 3.12
\end{itemize}

Computational requirements for full reproduction: approximately 30 minutes total wall-clock time on standard hardware plus API access to Anthropic (Claude), Groq (Llama models), and optionally Google (Gemini). Estimated API cost is USD 2 to 5. The dataset and analysis scripts are released under Creative Commons BY 4.0.

\section{Ethical Considerations}

The IyawoBench v1.0 dataset was generated from statistical distributions of real patient encounters. No individual patient data appears in any vignette. All patient identifiers, geographic identifiers below the state level, and temporal identifiers below the month level were excluded from the source distributions before vignette generation. The synthetic vignettes therefore carry no patient re-identification risk.

Our findings identify systematic failure modes in currently deployed clinical AI systems. We believe responsible disclosure requires that these findings be published so that deploying organisations can implement appropriate mitigations. The framework proposed here provides a diagnostic mechanism; it does not encourage or enable adversarial use of clinical AI systems.

We caution against interpreting our findings as recommendations against clinical AI deployment in low and middle income countries. The comparator in such settings is often no clinical decision support at all, delivered by Community Health Extension Workers with limited training. Even models exhibiting formal failure modes may substantially outperform this baseline. Our framework is intended to support informed deployment decisions, not to discourage deployment.

\section*{Acknowledgements}

We thank the clinical teams at 19 primary health centres in Oyo State, Nigeria for the deployment data underlying IyawoBench. We thank the Iyawo Health engineering team for infrastructure support. This work was funded by internal Iyawo Health research budget.

\bibliographystyle{plainnat}
\bibliography{references}

\end{document}